\documentclass{llncs}

\usepackage{amsmath,amssymb}
\usepackage{hyperref}         
\usepackage{thmtools}         
\usepackage{booktabs}         
\usepackage[inline]{enumitem} 
\usepackage{thm-restate}      
\setlist[enumerate]{label=(\arabic*)}
\usepackage{subfig}           
\usepackage{float}            
\usepackage[version=0.96]{pgf}
\usepackage{tikz}             
\usetikzlibrary{positioning}
\usetikzlibrary{arrows,automata}
\usetikzlibrary{calc}
\usetikzlibrary{petri}
\usetikzlibrary{decorations.pathreplacing}
\usetikzlibrary{decorations.pathmorphing}
\usetikzlibrary{trees}
\usetikzlibrary{shapes,shapes.geometric,fit}
\usetikzlibrary{backgrounds}
\usepackage{gnuplot-lua-tikz} 
\usepackage{macros}          

\hypersetup{colorlinks,linkcolor={black},citecolor={black},urlcolor={black}}

\begin{document}

\pagestyle{headings}

\title{Computing the concurrency threshold of sound free-choice workflow nets}

\author{\
Philipp J. Meyer\inst{1} \and
Javier Esparza\inst{1} \and
Hagen V\"olzer\inst{2}}

\institute{Technical University of Munich, Germany\\
\email{\{meyerphi,esparza\}@in.tum.de}
\and
IBM Research, Zurich, Switzerland\\
\email{hvo@zurich.ibm.com}
}

\maketitle              

\begin{abstract}
Workflow graphs extend classical flow charts with concurrent fork and join nodes.
They constitute the core of business processing languages such as BPMN 
or UML Activity Diagrams. The activities of a workflow graph are executed by
humans or machines, generically called resources. If concurrent activities cannot
be executed in parallel by lack of resources, the time needed to execute the workflow increases.
We study the problem of computing the minimal number of resources necessary to fully
exploit the concurrency of a given workflow, and execute it as fast as possible (i.e., as fast as with unlimited resources).

We model this problem using free-choice Petri nets, which are known to be equivalent to workflow graphs.
We analyze the computational complexity of two versions of the problem: computing the resource and concurrency thresholds.
We use the results to design an algorithm to approximate the concurrency threshold, and evaluate it on a
benchmark suite of 642 industrial examples. We show that it performs very well in practice: It always provides the 
exact value, and never takes more than 30 milliseconds for any workflow, even for those with a huge number of reachable markings.
\end{abstract}

\section{Introduction}\label{sec:introduction}
A \emph{workflow graph} is a classical control-flow graph (or flow chart) extended with concurrent fork and join. Workflow graphs represent the core of workflow languages such as BPMN (Business Process Model and Notation), EPC (Event-driven Process Chain), or UML Activity Diagrams. 

In many applications, the activities of an execution workflow graph have to be carried out by a fixed number
of \emph{resources} (for example, a fixed number of computer cores). Increasing the number of cores can reduce the minimal runtime of the workflow. For example, consider a simple deterministic workflow (a workflow without choice or merge nodes), which forks into $k$ parallel activities, all of duration 1, and terminates after a join. With an optimal assignment of resources to activities, the workflow takes time $k$ when executed with one resource, time $\lceil k/2 \rceil$ with two resources, and time $1$ with $k$ resources; additional resources bring no further reduction. We call $k$ the \emph{resource threshold}. In a deterministic workflow that forks into two parallel chains of $k$ sequential activities each, one resource leads to runtime $2k$, and two resources to runtime $k$. More resources do not improve the runtime, and so the resource threshold is 2. Clearly, the resource threshold of a deterministic workflow with $k$ activities is a number between $1$ and $k$.
Determining this number can be seen as a scheduling problem.
However, most scheduling problems assume a fixed number of resources and study how to
optimize the makespan~\cite{HallS96,Pinedo16}, while we study how to minimize the number of resources.
Other works on resource/machine minimization~\cite{ChuzhoyGKN04,ChuzhoyC09}
consider interval constraints instead of the partial-order constraints given by a workflow graph.

\begin{figure}
\begin{center}
\subfloat[Sound free-choice workflow net $N$]{\
    \begin{tikzpicture}[scale=0.7,every node/.style={scale=0.8}]
        \node [place,label={[xshift=0mm,yshift=-0.5mm] above:$i$}] (i) at (0,2) {\ttime{0}};
        \node [place,label={[xshift=0mm,yshift=-0.5mm] above:$p_1$}] (p1) at (2.25,3) {\ttime{1}};
        \node [place,label={[xshift=0mm,yshift=0.5mm] below:$p_2$}] (p2) at (2.25,1) {\ttime{1}};
        \node [place,label={[xshift=0mm,yshift=-0.5mm] above:$p_3$}] (p3) at (4.25,3) {\ttime{1}};
        \node [place,label={[xshift=0mm,yshift=0.5mm] below:$p_4$}] (p4) at (4.25,1) {\ttime{2}};
        \node [place,label={[xshift=0mm,yshift=-0.5mm] above:$p_5$}] (p5) at (3.25,2) {\ttime{2}};
        \node [place,label={[xshift=0.7mm,yshift=0mm] left:$p_6$}] (p6) at (5.5,1) {\ttime{1}};
        \node [place,label={[xshift=0mm,yshift=0.5mm] below:$p_7$}] (p7) at (4.25,0) {\ttime{1}};
        \node [place,label={[xshift=0mm,yshift=-0.5mm] above:$p_8$}] (p8) at (6.5,2) {\ttime{2}};
        \node [place,label={[xshift=0mm,yshift=0.5mm] below:$p_9$}] (p9) at (6.5,0) {\ttime{2}};
        \node [place,label={[xshift=0mm,yshift=-0.5mm] above:$o$}] (o) at (7.5,1) {\ttime{0}};

        \node [transition,label={[xshift=0mm,yshift=-0.5mm] above:$t_1$}] (t1) at (1,2) {}
          edge [pre]  (i)
          edge [post] (p1)
          edge [post] (p2);
        \draw [post] (t1) |- (p7);
        \node [transition,label={[xshift=0mm,yshift=-0.5mm] above:$t_2$}] (t2) at (3.25,3) {}
          edge [pre]  (p1)
          edge [post] (p3);
        \node [transition,label={[xshift=0mm,yshift=0.5mm] below:$t_3$}] (t3) at (3.25,1) {}
          edge [pre]  (p2)
          edge [post] (p4);
        \node [transition,label={[xshift=-0.7mm,yshift=0mm] right:$t_4$}] (t4) at (4.25,2) {}
          edge [pre]  (p3)
          edge [pre]  (p4)
          edge [post] (p5)
          edge [post] (p2);
        \node [transition,label={[xshift=0.7mm,yshift=0mm] left:$t_5$}] (t5) at (2.25,2) {}
          edge [pre]  (p5)
          edge [post] (p1);
        \node [transition,label={[xshift=0mm,yshift=-0.5mm] above:$t_6$}] (t6) at (5.5,2) {}
          edge [pre]  (p3)
          edge [pre]  (p4)
          edge [post] (p6)
          edge [post] (p8);
        \node [transition,label={[xshift=0mm,yshift=0.5mm] below:$t_7$}] (t7) at (5.5,0) {}
          edge [pre]  (p6)
          edge [pre]  (p7)
          edge [post] (p9);
        \node [transition,label={[xshift=0.7mm,yshift=0mm] left:$t_8$}] (t8) at (6.5,1) {}
          edge [pre]  (p8)
          edge [pre]  (p9)
          edge [post] (o);
    \end{tikzpicture}\label{fig:example-net}
}
\subfloat[A run of $N$]{\
    \begin{tikzpicture}[scale=0.7,every node/.style={scale=0.8}]
        \node [place,label={[xshift=0mm,yshift=-0.5mm] above:$i$}] (i) at (0,2) {\ttime{0}};
        \node [place,label={[xshift=0mm,yshift=-0.5mm] above:$p_1$}] (p1) at (2.25,3) {\ttime{1}};
        \node [place,label={[xshift=0mm,yshift=0.5mm] below:$p_2$}] (p2) at (2.25,1) {\ttime{1}};
        \node [place,label={[xshift=0mm,yshift=-0.5mm] above:$p_3$}] (p3) at (4.25,3) {\ttime{1}};
        \node [place,label={[xshift=0mm,yshift=0.5mm] below :$p_4$}] (p4) at (4.25,1) {\ttime{2}};
        \node [place,label={[xshift=0.7mm,yshift=0mm] left:$p_6$}] (p6) at (5.5,1) {\ttime{1}};
        \node [place,label={[xshift=0mm,yshift=0.5mm] below:$p_7$}] (p7) at (4.25,0) {\ttime{1}};
        \node [place,label={[xshift=0mm,yshift=-0.5mm] above:$p_8$}] (p8) at (6.5,2) {\ttime{2}};
        \node [place,label={[xshift=0mm,yshift=0.5mm] below:$p_9$}] (p9) at (6.5,0) {\ttime{2}};
        \node [place,label={[xshift=0mm,yshift=-0.5mm] above:$o$}] (o) at (7.5,1) {\ttime{0}};

        \node [transition,label={[xshift=0mm,yshift=-0.5mm] above:$t_1$}] (t1) at (1,2) {}
          edge [pre]  (i)
          edge [post] (p1)
          edge [post] (p2);
        \draw [post] (t1) |- (p7);
        \node [transition,label={[xshift=0mm,yshift=-0.5mm] above:$t_2$}] (t2) at (3.25,3) {}
          edge [pre]  (p1)
          edge [post] (p3);
        \node [transition,label={[xshift=0mm,yshift=0.5mm] below:$t_3$}] (t3) at (3.25,1) {}
          edge [pre]  (p2)
          edge [post] (p4);
        \node [transition,label={[xshift=0mm,yshift=-0.5mm] above:$t_6$}] (t6) at (5.5,2) {}
          edge [pre]  (p3)
          edge [pre]  (p4)
          edge [post] (p6)
          edge [post] (p8);
        \node [transition,label={[xshift=0mm,yshift=0.5mm] below:$t_7$}] (t7) at (5.5,0) {}
          edge [pre]  (p6)
          edge [pre]  (p7)
          edge [post] (p9);
        \node [transition,label={[xshift=0.7mm,yshift=0mm] left:$t_8$}] (t8) at (6.5,1) {}
          edge [pre]  (p8)
          edge [pre]  (p9)
          edge [post] (o);
    \end{tikzpicture}\label{fig:example-net-run}
}
\caption{A sound free-choice workflow net and one of its runs}\label{fig:intro-example}
\end{center}
\end{figure}
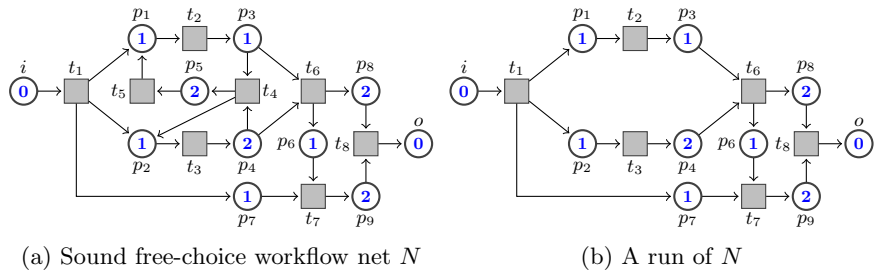

Following previous work, we do not directly work with workflow graphs, but with their equivalent representation as \emph{free-choice workflow Petri nets}, which has been shown to be essentially the same model~\cite{FavreFV15} and allows us to directly use a wealth of results of free-choice Petri nets~\cite{DeselEsparza95}. Fig.~\ref{fig:example-net} shows a free-choice workflow net. The actual workflow activities, also called \emph{tasks}, which need a resource to execute and which consume time are modeled as the places of the net: Each place $p$ of the net is assigned a time $\tau(p)$, depicted in blue. Intuitively, when a token arrives in $p$, it must execute a task that takes $\tau(p)$ time units before it can be used to fire a transition. A free choice exists between transitions $t_4$ and $t_6$, which is a representation of a choice node (if-then-else or loop condition) in the workflow. 

If no choice is present or all choices are resolved, we have a deterministic workflow such as the one in Fig.~\ref{fig:example-net-run}. In Petri net terminology, deterministic workflows correspond to the class of marked graphs. Deterministic workflows are common in practice: in the standard suite of 642 industrial workflows that we use for experiments, $63.7\%$ percent are deterministic. We show that already for this restricted class, deciding if the threshold exceeds a given bound is NP-hard.  Therefore, we investigate an over-approximation of the resource threshold, already introduced in~\cite{BVT16}: the \emph{concurrency threshold}. This is the maximal number of task places that can be simultaneously marked at a reachable marking. Clearly, if a workflow with concurrency threshold $k$ is executed with $k$ resources, then we can always start the task of a place immediately after a token arrives, and this schedule already achieves the fastest runtime achievable with unlimited resources. We show that the concurrency threshold can be computed in polynomial time for deterministic workflows.

For workflows with nondeterministic choice, corresponding to free-choice nets, we show that 
computing the concurrency threshold of free-choice workflow nets is NP-hard, solving a problem left open in~\cite{BVT16}. We even prove that the problem remains NP-hard for sound free-choice workflows. Soundness is the dominant behavioral correctness notion for workflows, which rules out basic control-flow errors such as deadlocks. NP-hardness in the sound case is remarkable, because many analysis problems that have high complexity in the unsound case can be solved in polynomial time in the sound case (see e.g.~\cite{DeselEsparza95,Aalst97,DBLP:conf/qest/EsparzaHS16}).

After our complexity analysis, we design an algorithm to compute bounds on the concurrency threshold using a combination of linear optimization and state-space exploration.
We evaluate it on a benchmark suite of 642 sound free-choice workflow nets from an industrial source (IBM)~\cite{FFJKLVW09}. The bounds can be computed in a total of 7 seconds (over all 642 nets). On the contrary, the computation of the exact value by state-space exploration techniques times out for the three largest nets, and takes 7 minutes for the rest. (Observe that partial-order reduction techniques cannot be used, because one may then miss the interleaving realizing the concurrency threshold.)

The paper is structured as follows. Section~\ref{sec:preliminaries} contains preliminaries.
Sections~\ref{sec:rt} and~\ref{sec:ct} study the resource and concurrency thresholds, respectively.
Section~\ref{sec:experiments} presents our algorithms for computing the concurrency bound, and experimental results.
Finally, Section~\ref{sec:conclusion} contains conclusions.

\section{Preliminaries}\label{sec:preliminaries}
\subsubsection{Petri nets.} A \emph{Petri net} $N$ is a tuple $(P, T, F)$ where
$P$ is a finite set of places, $T$ is a finite set of transitions
($P \cap T = \emptyset$), and $F \subseteq (P \times T) \cup (T \times P)$
is a set of arcs. The \emph{preset} of $x \in P \cup T$
is $\preset{x} \defeq \set{ y \mid (y, x) \in F }$ and its
\emph{postset} is $\postset{x} \defeq \set{ y \mid (x, y) \in F }$.
We extend the definition of presets and postsets to sets of places and 
transitions $X \subseteq P \cup T$ by $\preset{X} \defeq \bigcup_{x \in X} \preset{x}$
and $\postset{X} \defeq \bigcup_{x \in X} \postset{x}$.
A net is \emph{acyclic} if the relation $F^*$ is a partial order,
denoted by $\preceq$ and called the \emph{causal order}.
A node $x$ of an acyclic net is \emph{causally maximal} if no node y satisfies $x \prec y$.

A \emph{marking} of a Petri net is a function $M \colon P \to \N$, representing
the number of tokens in each place. For a set of places $S \subseteq P$, we define 
$M(S) \defeq \sum_{p \in S} M(p)$. Further, for a set of places $S \subseteq P$, we 
define by $M_S$ the marking with $M_S(p) = 1$ for $p \in S$ and
$M_S(p) = 0$ for $p \notin S$.

A transition $t$ is \emph{enabled} at a marking $M$ if
for all $p \in \preset{t}$, we have $M(p) \ge 1$.
If $t$ is enabled at $M$, it may \emph{occur}, leading to a marking $M'$ obtained by 
removing one token from each place of $\preset{t}$  and then adding one token to 
each place of $\postset{t}$.
We denote this by $M \trans{t} M'$.
Let $\sigma = t_1 t_2 \ldots t_n$ be a sequence of transitions.
For a marking $M_0$, $\sigma$ is an \emph{occurrence sequence} if
$M_0 \trans{t_1} M_1 \trans{t_2} \ldots \trans{t_n} M_n$
for some markings $M_1,\ldots,M_n$. We say that $M_n$ is reachable
from $M_0$ by $\sigma$ and denote this by $M_0 \trans{\sigma} M_n$.
The set of all markings reachable from $M$ in $N$ by some occurrence
sequence $\sigma$ is denoted by $\reach[N]{M}$. A \emph{system} is a pair $(N, M)$ of a Petri net $N$ and a marking $M$.
A system $(N, M)$ is \emph{live} if for every $M' \in \reach[N]{M}$ and every transition $t$ some marking $M'' \in \reach[N]{M'}$ enables $t$. The system  is \emph{1-safe} if $M'(p) \leq 1$ for every $M' \in \reach[N]{M}$ and every place $p \in P$.\\
\noindent {\bf Convention:} Throughout this paper we assume that systems are 1-safe, i.e., we identify
``system'' and ``1-safe system''.

\subsubsection{Net classes.} A net $N=(P,T,F)$ is 
a \emph{marked graph} if $|\preset{p}| \leq 1$ and $|\postset{p}| \leq 1$ for every place $p \in P$, and a \emph{free-choice net} if for any two places $p_1, p_2 \in P$ either $\postset{p_1} \cap \postset{p_2} = \emptyset$ or $\postset{p_1} = \postset{p_2}$.

\subsubsection{Non-sequential processes of Petri nets.}
An \emph{$(A,B)$-labeled Petri net} is a tuple $N=(P, T, F, \lambda, \mu)$, where 
$\lambda \colon P \rightarrow A$ and $\mu \colon T \rightarrow B$ are \emph{labeling functions} over alphabets $A, B$. 
The nonsequential processes of a 1-safe system $(N, M)$ are acyclic, $(P,T)$-labeled marked graphs.  Say that a set $P''$ of places of a $(P,T)$-labeled acyclic net \emph{enables} $t \in T$ if all the places of $P''$ are causally maximal, carry pairwise distinct labels, and $\lambda(P'')=\preset{t}$. 

\begin{definition}
Let $N=(P,T,F)$ be a Petri net and let $M$ be a marking of $N$.
The set $\NP{N, M}$ of \emph{nonsequential processes} of $(N, M)$ (\emph{processes} for short) is the set of $(P,T)$-labeled Petri nets 
defined inductively as follows:
\begin{itemize}
\item The $(P,T)$-labeled Petri net containing for each place $p \in P$ marked at $M$ one place $\widehat{p}$ labeled by $p$, no other places, and no transitions, belongs to $\NP{N, M}$.
\item If $\Pi=(P', T', F', \lambda, \mu) \in \NP{N, M}$ and $P'' \subseteq P'$ enables some transition $t$ of $N$, then the $(P,T)$-labeled net $\Pi_t =(P' \uplus \widehat{P}, T' \uplus \{\,\widehat{t}\,\}, F' \uplus \widehat{F}, \lambda \uplus \widehat{\lambda}, \mu \uplus \widehat{\mu})$, where
\begin{itemize}
\item $\widehat{P} = \{\,\widehat{p} \mid p \in \postset{t} \}$, with $\widehat{\lambda}(\,\widehat{p}\,)=p$, and 
$\widehat{\mu}(\,\widehat{t}\,)=t$; 
\item $\widehat{F} = \{ (\,p'', \widehat{t}\,) \mid p'' \in P'' \} \cup \{ (\,\widehat{t}, \widehat{p}\,) \mid \widehat{p} \in \widehat{P}\}$;
\end{itemize}
\noindent also belongs to $\NP{N, M}$. We say that $\Pi_t$ \emph{extends} $\Pi$.
\end{itemize}
We denote the minimal and maximal places of a process $\Pi$ w.r.t.~the causal order by $\min(\Pi)$ and $\max(\Pi)$, respectively. 
\end{definition}

As usual, we say that two processes are \emph{isomorphic} if they are the same up to renaming of the places and transitions 
(notice that we rename only the names of the places and transitions, not their labels).

Fig.~\ref{fig:example-process} shows two  processes of the workflow net in Fig.~\ref{fig:example-net}. (The figure does not show the names of places and transitions, only their labels.) The net containing the white and grey nodes only is already a process, and the grey places are causally maximal places that enable $t_6$. Therefore, according to the definition we can extend the process with the green nodes to produce another process. On the right we extend the same process in a different way, with the transition $t_4$.

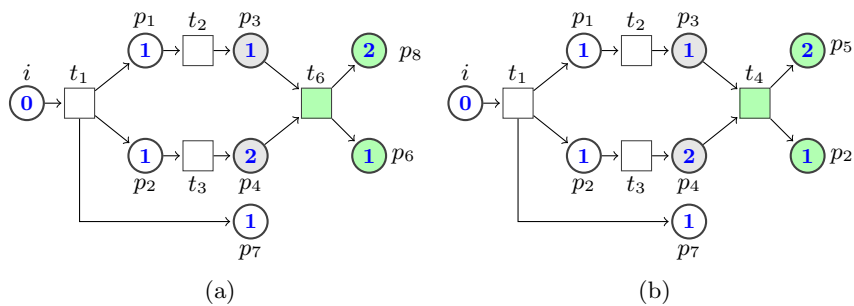
\begin{figure}
\begin{center}
\subfloat[]{\
    \begin{tikzpicture}[scale=0.7,every node/.style={scale=1}]
        \node [place,label={[xshift=0mm,yshift=-0.5mm] above:$i$}] (i) at (0,2) {\ttime{0}};
        \node [place,label={[xshift=0mm,yshift=-0.5mm] above:$p_1$}] (p1) at (2.25,3) {\ttime{1}};
        \node [place,label={[xshift=0mm,yshift=0.5mm] below:$p_2$}] (p2) at (2.25,1) {\ttime{1}};
        \node [place, fill=black!10, label={[xshift=0mm,yshift=-0.5mm] above:$p_3$}] (p3) at (4.25,3) {\ttime{1}};
        \node [place,fill=black!10, label={[xshift=0mm,yshift=0.5mm] below:$p_4$}] (p4) at (4.25,1) {\ttime{2}};
        \node [place,fill=green!30,label={[xshift=-0.5mm,yshift=0mm] right:$p_6$}] (p6) at (6.5,1) {\ttime{1}};
        \node [place,label={[xshift=0mm,yshift=0.5mm] below:$p_7$}] (p7) at (4.25,-0.25) {\ttime{1}};
        \node [place,fill=green!30,label={[xshift=0.5mm,yshift=-0.5mm] right:$p_8$}] (p8) at (6.5,3) {\ttime{2}};

        \node [transition,fill=white!25,label={[xshift=0mm,yshift=-0.5mm] above:$t_1$}] (t1) at (1,2) {}
          edge [pre]  (i)
          edge [post] (p1)
          edge [post] (p2);
        \draw [post] (t1) |- (p7);
        \node [transition,fill=white!25,label={[xshift=0mm,yshift=-0.5mm] above:$t_2$}] (t2) at (3.25,3) {}
          edge [pre]  (p1)
          edge [post] (p3);
        \node [transition,fill=white!25,label={[xshift=0mm,yshift=0.5mm] below:$t_3$}] (t3) at (3.25,1) {}
          edge [pre]  (p2)
          edge [post] (p4);
        \node [transition,fill=green!30,label={[xshift=0mm,yshift=-0.5mm] above:$t_6$}] (t6) at (5.5,2) {}
          edge [pre]  (p3)
          edge [pre]  (p4)
          edge [post] (p6)
          edge [post] (p8);
    \end{tikzpicture}\label{fig:example-process-1}
}
\subfloat[]{\
    \begin{tikzpicture}[scale=0.7,every node/.style={scale=1}]
        \node [place,label={[xshift=0mm,yshift=-0.5mm] above:$i$}] (i) at (0,2) {\ttime{0}};
        \node [place,label={[xshift=0mm,yshift=-0.5mm] above:$p_1$}] (p1) at (2.25,3) {\ttime{1}};
        \node [place,label={[xshift=0mm,yshift=0.5mm] below:$p_2$}] (p2) at (2.25,1) {\ttime{1}};
        \node [place, fill=black!10, label={[xshift=0mm,yshift=-0.5mm] above:$p_3$}] (p3) at (4.25,3) {\ttime{1}};
        \node [place,fill=black!10, label={[xshift=0mm,yshift=0.5mm] below:$p_4$}] (p4) at (4.25,1) {\ttime{2}};
        \node [place,fill=green!30,label={[xshift=-0.5mm,yshift=0mm] right:$p_2$}] (p6) at (6.5,1) {\ttime{1}};
        \node [place,label={[xshift=0mm,yshift=0.5mm] below:$p_7$}] (p7) at (4.25,-0.25) {\ttime{1}};
        \node [place,fill=green!30,label={[xshift=-0.5mm,yshift=0.5mm] right:$p_5$}] (p8) at (6.5,3) {\ttime{2}};
        \node [transition,fill=white!25,label={[xshift=0mm,yshift=-0.5mm] above:$t_1$}] (t1) at (1,2) {}
          edge [pre]  (i)
          edge [post] (p1)
          edge [post] (p2);
        \draw [post] (t1) |- (p7);
        \node [transition,fill=white!25,label={[xshift=0mm,yshift=-0.5mm] above:$t_2$}] (t2) at (3.25,3) {}
          edge [pre]  (p1)
          edge [post] (p3);
        \node [transition,fill=white!25,label={[xshift=0mm,yshift=0.5mm] below:$t_3$}] (t3) at (3.25,1) {}
          edge [pre]  (p2)
          edge [post] (p4);
        \node [transition,fill=green!30,label={[xshift=0mm,yshift=-0.5mm] above:$t_4$}] (t6) at (5.5,2) {}
          edge [pre]  (p3)
          edge [pre]  (p4)
          edge [post] (p6)
          edge [post] (p8);
    \end{tikzpicture}\label{fig:example-process-2}
}
\caption{Nonsequential processes of the net of Fig.~\ref{fig:example-net}}\label{fig:example-process}
\end{center}
\end{figure}

The following is well known. Let  $(P', T', F', \lambda, \mu)$ be a process of $(N, M)$:
\begin{itemize}
\item For every linearization $\sigma = t_1' \ldots t_n'$ of $T'$ respecting the causal order $\preceq$, the sequence $\mu(\sigma) = \mu(t_1') \ldots \mu(t_n')$ is a firing sequence of $(N, M)$. Further, all these firing sequences lead to the same marking. We call it the \emph{final marking} of $\Pi$, and say that $\Pi$ leads from $M$ to its final marking. \\
For example, in Fig.~\ref{fig:example-process} the sequences of the right process labeled by $t_1 t_2 t_3 t_4$ 
and $t_1 t_3 t_2 t_4$  are firing sequences leading to the marking $M=\{p_2,p_5,p_7\}$. 
\item For every firing sequence $t_1 \cdots t_n$ of $(N, M)$ there is a  process 
$(P', T', F', \lambda, \mu)$ such that $T'=\{t_1', \ldots, t_n'\}$, $\mu(t_i') = t_i$ for every 
$1 \leq i \leq n$, and  $\mu(t_i') \preceq \mu(t_j')$ implies $i \leq j$.
\end{itemize}

\subsubsection{Workflow nets.} We slightly generalize the definition of workflow net as presented in e.g.~\cite{Aalst97} by allowing multiple
initial and final places. 
A \emph{workflow} net is a Petri net with two distinguished sets
$I$ and $O$ of \emph{input places} and \emph{output places} such that
(a) $\preset{I} = \emptyset = \postset{O}$ and
(b) for all $x \in P \cup T$, there exists a path
from some $i \in I$ to some $o \in O$ passing through $x$. The markings $M_I$ and $M_O$ are called initial and final markings of $N$.  A workflow net $N$ is \emph{sound} if
\begin{itemize}
    \item $\forall M \in \reach[N]{M_I} : M_{O} \in \reach[N]{M}$,
    \item $\forall M \in \reach[N]{M_{I}} : (M(O) \ge |O|) \Rightarrow (M = M_O)$, and
    \item $\forall t \in T : \exists M \in \reach[N]{M_I} : t$ is enabled at $M$.
\end{itemize}
It is well-known that every sound free-choice workflow net is a 1-safe system with the initial marking $M_I$~\cite{Aalst00,DeselEsparza95}.
Given a workflow net according to this definition one can construct another one with one 
single input place $i$ and output place $o$ and two transitions $t_i, t_o$ with 
$\preset{t_i} = \{i\}, \postset{t_i}=I$ and $\preset{t_o} = O, \postset{t_o}=\{o\}$. For 
all purposes of this paper these two workflow nets are equivalent.

Given a workflow net $N$, we say that a process $\Pi$ of $(N, M_I)$ is a \emph{run} if it leads to $M_O$. 
For example, the net in Fig.~\ref{fig:example-net-run} is a run of the net in Fig.~\ref{fig:example-net}. 

\subsubsection{Petri nets with task durations}  

We consider Petri nets in which, intuitively, when a token arrives in a place $p$ it has to execute a task taking $\tau(p)$ time units before the token can be used to fire any transition. Formally, we consider tuples $N=(P,T,F,\tau)$ where $(P,T,F)$ is a net and $\tau \colon P \rightarrow \N$. 

\begin{definition}\label{def:resbound}
Given a nonsequential process $\Pi = (P',T',F',\lambda,\mu)$ of $(N, M)$, a time bound $t$, and a number of resources $k$, we say that \emph{$\Pi$ is executable within time $t$ with $k$ resources} if there is a function $f \colon P' \rightarrow \N$ such that
\begin{itemize}
\item[(1)] for every $p_1', p_2' \in P'$: if $p_1' \prec p_2'$ then $f(p_1') + \tau(\lambda(p_1')) \leq f(p_2')$;
\item[(2)] for every $p' \in P'$: $f(p') + \tau(\lambda(p')) \leq t$; and 
\item[(3)] for every $0 \leq u < t$ there are at most $k$ places $p' \in P'$ such that $f(p') \leq u < f(p') + \tau(p')$.
\end{itemize}
\noindent We call a function $f$ satisfying (1) a \emph{schedule}, a function satisfying (1) and (2) a \emph{$t$-schedule}, and a function satsifying (1)-(3) a \emph{$(k,t)$-schedule} of $\Pi$. 
\end{definition}
Intuitively, $f(p')$ describes the starting time of the task executed at $p'$. Condition (1)
states that if  $p_1' \preceq p_2'$, then the task associated to $p_2'$ can only start after the task for $p_1'$ has ended; condition (2) states that all tasks are done by time $t$, and condition (3) that at any moment in time at most $k$ tasks are being executed. As an example, the process in Fig.~\ref{fig:example-net-run} can be executed with two resources in time $6$ with the schedule  $i, p_1, p_2 \mapsto 0$; $p_3, p_4 \mapsto 1$; $ p_7, p_6 \mapsto 3$, and $p_8, p_9 \mapsto 4$. 

Given a process $\Pi = (P',T',F',\lambda,\mu)$ of $(N, M)$ we define the schedule  $f_{\min}$ as follows: if $p' \in \min(\Pi)$ then $f_{\min}(p')=0$,
otherwise define $f_{\min}(p') = \max\{f_{\min}(p'') + \tau(\lambda(p'')) \mid p'' \preceq  p' \}$. Further, we define
the \emph{minimal execution time} $t_{\min}(\Pi) =\max\{ f(p') + \tau(\lambda(p'')) \mid p' \in \max(\Pi) \}$. In the process in Fig.~\ref{fig:example-net-run}, the schedule $f_{\min}$ is the function that assigns $i, p_1, p_2, p_7~\mapsto~0$, $p_3, p_4~\mapsto~1$, $p_6, p_8~\mapsto~3$, $p_9~\mapsto~4$, and $o~\mapsto~6$, and so $t_{\min}(\Pi)= 6$. We have:

\begin{lemma}
A process $\Pi = (P',T',F',\lambda,\mu)$ of $(N, M)$ can be executed within time $t_{\min}(\Pi)$ with $|P'|$ resources, and cannot be executed faster with any number of resources.
\end{lemma}
\begin{proof}
For $k \geq |P'|$ resources condition (3) of Definition~\ref{def:resbound} holds vacuously.
$\Pi$ is executable within time $t$ if{}f conditions (1) and (2) hold. Since $f_{\min}$ satisfies (1) and (2)  for $t = t_{\min}(\Pi)$,
$\Pi$ can be executed within time $t_{\min}(\Pi)$. Further, $t_{\min}(\Pi)$ is the smallest time for which (1) and (2) can hold,
and so $\Pi$ cannot be executed faster with any number of resources.
\end{proof}

\section{Resource threshold}\label{sec:rt}
We define the resource threshold of a run of a workflow net, and of the net itself.  Intuitively, the resource threshold of a run is the minimal number of resources that allows one to execute it as fast as with unlimited resources, and the resource threshold of a workflow net  is the minimal number of resources that allows one to execute \emph{every run} as fast as with unlimited resources.

\begin{definition} 
Let $N$ be a workflow net, and let $\Pi$ be a run of $N$. The \emph{resource threshold of $\Pi$}, denoted by $\RT{\Pi}$ is the smallest number $k$ such that $\Pi$ can be executed in time $t_{\min}(\Pi)$ with $k$ resources.  A schedule of $\Pi$ \emph{realizes} the resource threshold if it is a $(\RT{\Pi}, t_{\min}(\Pi))$-schedule. 

The \emph{resource threshold} of $N$, denoted by $\RT{N}$, is defined by $\RT{N} = \max\{ \RT{\Pi} \mid \Pi \mbox{ is a run of $(N, M_I)$} \}$. A \emph{schedule of $N$} is a function that assigns to every process $\Pi \in {\cal NP}(N, M)$ a schedule of $\Pi$. A schedule of $N$ is a $(k, t)$-schedule if it assigns to every run $\Pi$ a $(k, t)$-schedule of $\Pi$. A schedule of $N$ \emph{realizes} the resource threshold if it assigns to every run $\Pi$ a $(\RT{N}, t_{\min}(\Pi))$-schedule. 
\end{definition}

\begin{example}
    We have seen in the previous section that for the process in Fig.~\ref{fig:example-net-run} we have $t_{\min}(\Pi)= 6$, and a schedule with two resources already achieves this time. So the resource bound is $2$. The workflow net of Fig.~\ref{fig:intro-example} has infinitely many runs, in which 
loosely speaking, the net executes $t_4$ arbitrarily many times, until it ``exits the loop'' by choosing $t_6$, followed by $t_7$ and $t_8$. It can be shown that all processes have resource threshold $2$, and so that is also the resource threshold of the net.
\end{example}

In the rest of the section we obtain two negative results about the result threshold. First, it is difficult to compute: Determining if the resource threshold exceeds a given threshold is NP-complete even for acyclic marked graphs, a very simple class of workflows. 
Second, we show that even for acyclic free-choice workflow nets the resource threshold may not be realized by any online scheduler.

\subsection{Resource threshold is NP-complete for acyclic marked graphs}
We prove that deciding if the resource threshold exceeds a given bound is NP-complete even for acyclic sound marked graphs.
The proof proceeds by reduction from the following classical scheduling problem, proved NP-complete in~\cite{Ullmann75}:
\begin{quote}
{\bf Given}:  a finite, partially ordered set of jobs with non-negative integer durations, and non-negative integers $t$ and $k$.\\
{\bf Decide}: Can all jobs can be executed with $k$ machines within $t$ time units in a way that respects the given partial order, i.e.,
a job is started only after all its predecessors have been finished?
\end{quote}
More formally, the problem is defined as follows: Given jobs ${\cal J}=\{J_1, \ldots, J_n\}$, where $J_i$ has duration $\tau(J_i)$ for every $1 \leq i \leq n$, and a partial order $\preceq$ on ${\cal J}$, does there exist a function $f \colon {\cal J} \rightarrow \N$ such that
\begin{itemize}
\item[(1)] for every $1\leq i, j \leq n$: if $J_i \prec J_j$ then $f(J_i) + \tau(J_i) \leq f(J_j)$;
\item[(2)] for every $1 \leq i \leq n$: $f(J_i) + \tau(J_i) \leq t$; and 
\item[(3)] for every $0 \leq u < t$ there are at most $k$ indices $i$ such that $f(J_i) \leq u < f(J_i) + \tau(J_i)$.
\end{itemize}

These conditions are almost identical to the ones we used to define if a nonsequential process can be executed within time $t$ with $k$ resources. We exploit this to construct an acyclic workflow marked graph that ``simulates'' the scheduling problem. A proof of the following theorem can be found in the Appendix.

\begin{restatable}{theorem}{thmNPMG}\label{thm:NPMG}
The following problem is NP-complete:
\begin{quote}
{\bf Given}:  An acyclic, sound workflow marked graph $N$, and a number $k$. \\
{\bf Decide}:  Does $\RT{N} \leq k$ hold?
\end{quote}
\end{restatable}

\subsection{Acyclic free-choice workflow nets may have no optimal online schedulers}
A resource threshold of $k$ guarantees that every run \emph{can} be executed without penalty with $k$ resources. In other words, \emph{there exists} a schedule that achieves optimal runtime. However, in many applications the schedule must be determined at runtime, that is, the resources must be allocated without knowing how choices will be resolved in the future. In order to formalize this idea we define  the notion of an \emph{online schedule} of a workflow net $N$.

\begin{definition}
Let $N$ be a Petri net, and let $\Pi$ and $\Pi'$ be two processes of $(N, M)$. 
We say that $\Pi$ is a \emph{prefix} of $\Pi'$, denoted by $\Pi \lhd \Pi'$, if there is a sequence $\Pi_1, \ldots, \Pi_n$
of processes such that $\Pi_1=\Pi$, $\Pi_n=\Pi'$, and $\Pi_{i+1}$ extends $\Pi_i$ by one transition for every $1 \leq i \leq n-1$.

Let $f$ be a schedule of $(N, M)$, i.e.,~a function assigning a schedule to each process. We say that $f$ is an \emph{online schedule} if for every two runs $\Pi_1, \Pi_2$, and for every two prefixes $\Pi_1' \lhd \Pi_1$ and $\Pi_2' \lhd \Pi_2$: If $\Pi_1'$ and $\Pi_2'$ are isomorphic, then $f(\Pi_1') = f(\Pi_2')$.
\end{definition}

Intuitively, if $\Pi_1'$ and $\Pi_2'$ are isomorphic then they are the same process $\Pi$, which in the future can be extended to either $\Pi_1$ or $\Pi_2$, depending on which transitions occur. In an online schedule, $\Pi$ is scheduled in the same way,  independently of whether
it will become $\Pi_1$ or $\Pi_2$ in the future. We show that even 
for acyclic free-choice workflow nets there may be no online schedule that realizes the resource threshold. That is, even though for every run it is possible to schedule the tasks with $\RT{N}$ resources to achieve optimal runtime, this requires knowing how it will evolve before the execution of the workflow.

\begin{proposition}
There is an acyclic, sound free-choice workflow net for which no online schedule realizes the resource threshold.
\end{proposition}
\begin{proof}
Consider the sound free-choice workflow net $(N, M_I)$ of Fig.~\ref{fig:noonlinescheduler}.
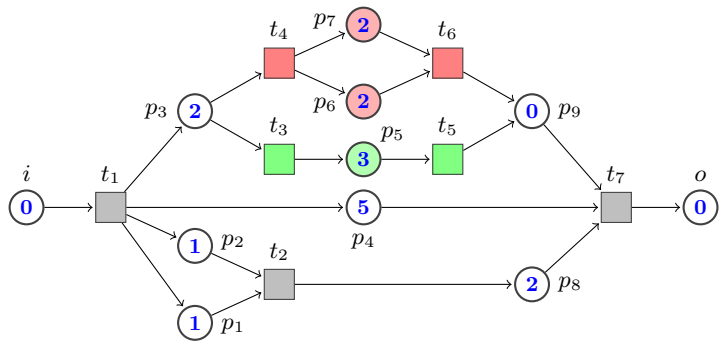
\begin{figure}
    \begin{center}
    \begin{tikzpicture}[scale=0.8,every node/.style={scale=1.0}]

\def\vs{0.8}
\def\hs{1.4}

            \node [place,label={[xshift=0mm,yshift=0mm] above:$i$}] (i) at (0*\hs,0) {\ttime{0}};
            \node [place,label={[xshift=0mm,yshift=-0.5mm] right:$p_1$}] (p1) at (2*\hs,-2.4*\vs) {\ttime{1}};
            \node [place,label={[xshift=0mm,yshift=0.5mm] right:$p_2$}] (p2) at (2*\hs,-0.8*\vs) {\ttime{1}};
            \node [place,label={[xshift=0mm,yshift=0mm] left:$p_3$}] (p3) at (2*\hs,2*\vs) {\ttime{2}};
            \node [place,label={[xshift=0mm,yshift=0mm] below:$p_4$}] (p4) at (4*\hs,0*\vs) {\ttime{5}};
            \node [place,fill=green!30,label={[xshift=-0.5mm,yshift=-0.5mm] above right:$p_5$}] (p5) at (4*\hs,1*\vs) {\ttime{3}};
            \node [place,fill=red!30,label={[xshift=0mm,yshift=-0.5mm] left:$p_6$}] (p6) at (4*\hs,2.2*\vs) {\ttime{2}};
            \node [place,fill=red!30,label={[xshift=0mm,yshift=0.5mm] left:$p_7$}] (p7) at (4*\hs,3.8*\vs) {\ttime{2}};
            \node [place,label={[xshift=0mm,yshift=0mm] right:$p_8$}] (p8) at (6*\hs,-1.6*\vs) {\ttime{2}};
            \node [place,label={[xshift=0mm,yshift=0mm] right:$p_9$}] (p9) at (6*\hs,2*\vs) {\ttime{0}};
            \node [place,label={[xshift=0mm,yshift=0mm] above:$o$}] (o) at (8*\hs,0*\vs) {\ttime{0}};

            \node [transition,label=above:$t_1$] (t1) at (1*\hs,0*\vs) {}
              edge [pre]  (i)
              edge [post] (p1) edge [post] (p2) edge [post] (p3) edge [post] (p4);
            \node [transition,label=above:$t_2$] (t2) at (3*\hs,-1.6*\vs) {}
              edge [pre]  (p1) edge [pre]  (p2)
              edge [post] (p8);
            \node [transition,fill=green!50, label=above:$t_3$] (t3) at (3*\hs,1*\vs) {}
              edge [pre]  (p3)
              edge [post] (p5);
            \node [transition,fill=red!50,label=above:$t_4$] (t4) at (3*\hs,3*\vs) {}
              edge [pre]  (p3)
              edge [post] (p6) edge [post] (p7);
            \node [transition,fill=green!50,label=above:$t_5$] (t5) at (5*\hs,1*\vs) {}
              edge [pre]  (p5)
              edge [post] (p9);
            \node [transition,fill=red!50, label=above:$t_6$] (t6) at (5*\hs,3*\vs) {}
              edge [pre]  (p6) edge [pre]  (p7)
              edge [post] (p9);
            \node [transition,label=above:$t_7$] (t7) at (7*\hs,0*\vs) {}
              edge [pre]  (p4) edge [pre]  (p8) edge [pre]  (p9)
              edge [post] (o);
        \end{tikzpicture}
    \end{center}
    \caption{A workflow net with two runs. No online scheduler for three resources 
    achieves the minimal runtime in both runs.}\label{fig:noonlinescheduler}
\end{figure}
It has two runs: $\Pi_g$, which executes the grey and green transitions,  and $\Pi_r$, which executes the grey and red transitions. Their resource thresholds are $\RT{\Pi_g}=\RT{\Pi_r}=3$, realized by the following schedules $f_g$ and $f_r$ in Fig.~\ref{fig:schedules}:

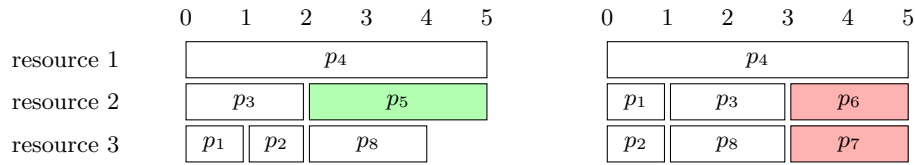
\begin{figure}[H]
\centering
\begin{tikzpicture}[scale=0.8,every node/.style={scale=1.0}]
\foreach \x in {0,1,...,5} {\node (n1) at (\x,1) {$\x$}; }; 
\node (r1) at (-2.0,0.3) {resource 1}; \draw (0,0) rectangle (5, 0.6) node[pos=.5]{$p_4$};
\node (r2) at (-2.0,-0.4) {resource 2};\draw (0,-0.7) rectangle (1.95, -0.1) node[pos=.5]{$p_3$}; \draw[fill=green!30] (2.05,-0.7) rectangle (5, -0.1) node[pos=.5]{$p_5$};
\node (r3) at (-2.0,-1.1) {resource 3};\draw (0,-1.4) rectangle (0.95, -0.8) node[pos=.5]{$p_1$}; \draw (1.05,-1.4) rectangle (1.95, -0.8) node[pos=.5]{$p_2$};
\draw (2.05,-1.4) rectangle (4, -0.8) node[pos=.5]{$p_8$};

\def\st{7}
\foreach \x in {0,1,...,5} {\node (n1) at (\x+\st,1) {$\x$}; }; 
\draw (0+\st,0) rectangle (5+\st, 0.6) node[pos=.5]{$p_4$};
\draw (0+\st,-0.7) rectangle (0.95+\st, -0.1) node[pos=.5]{$p_1$}; \draw (1.05+\st,-0.7) rectangle (2.95+\st, -0.1) node[pos=.5]{$p_3$};
\draw[fill=red!30] (3.05+\st,-0.7) rectangle (5+\st, -0.1) node[pos=.5]{$p_6$};
\draw (0+\st,-1.4) rectangle (0.95+\st, -0.8) node[pos=.5]{$p_2$}; \draw (1.05+\st,-1.4) rectangle (2.95+\st, -0.8) node[pos=.5]{$p_8$};
\draw[fill=red!30] (3.05+\st,-1.4) rectangle (5+\st, -0.8) node[pos=.5]{$p_7$};

\end{tikzpicture}
\caption{Schedules $f_g$ and $f_r$ for the two runs $\Pi_g$ and $\Pi_r$ of the net of Fig.~\ref{fig:noonlinescheduler}.}
\label{fig:schedules}
\end{figure}

Indeed, observe that $f_g$ and $f_r$ execute $\Pi_g$ and $\Pi_r$ within time $5$, and even
with unlimited resources no schedule can be faster because of the task $p_4$,
while two or fewer resources are insufficient to execute either run within time $5$.

The schedule of $(N, M_I)$ that assigns $f_g$ and $f_r$ to $\Pi_g$ and $\Pi_r$ is not an online schedule. Indeed, the process containing one single transition labeled by $t_1$ and places labeled by $i, p_1, p_2, p_3$ is isomorphic to prefixes of $\Pi_g$ and 
$\Pi_r$. However, we have $f_g(p_3) =0 \neq 1 = f_r(p_3)$. We now claim: 
\begin{itemize}
\item[(a)] Every schedule $f_g$ of $\Pi_g$ that realizes the resource threshold (time $5$ with 3 resources) satisfies $f_g(p_3)=0$.\\
Indeed, if $f_g(p_3)\geq 1$, then $f_g(p_5) \geq 3$, $f_g(p_9) \geq 6$, and finally $f_g(o)\geq 6$, so $f_g$ does not meet the time bound.
\item[(b)] Every schedule $f_r$ of $\Pi_r$ that realizes the resource threshold (time $5$ with 3 resources) satisfies $f_r(p_3) > 0$.\\
Observe first that we necessarily have $f_r(p_4)=0$, and so a resource, say $R_1$, is bound to $p_4$ during the complete execution of the workflow, leaving two resources left. Assume $f_r(p_3) = 0$, i.e., a second resource, say $R_2$, is bound to $p_3$ at time $0$, leaving one resource left, say $R_3$. Since both $p_1$ and $p_2$ must be executed before $p_8$, and only $R_3$ is free until time $2$, we get $f_r(p_8) \geq 2$.  So at time $2$ we still have to execute $p_6, p_7, p_8$ with resources $R_2, R_3$. Therefore, two out of $p_6, p_7, p_8$ must be executed sequentially by the same resource. Since $p_6, p_7, p_8$ take $2$ time units each,  one of the two resources needs time $4$, and we get $f_r(o) \geq 6$. 
\end{itemize}
By this claim, at time $0$, an online schedule has to decide whether to allocate a resource to $p_3$ or not, without knowing which of $t_3$ or $t_4$ will be executed in the future. If it schedules $f(p_3)=0$ and later $t_4$ occurs, then $\Pi_r$ is executed and the deadline of 5 time units is not met. The same occurs if it schedules $f(p_3)>0$, and later $t_3$ occurs.  
\end{proof}

\section{Concurrency threshold}\label{sec:ct}
Due to the two negative results presented in the previous section, we study a different parameter,
introduced in~\cite{BVT16}, called the concurrency threshold.
During execution of a business process, information on the resolution of future choices is often not available, and further no information on the possible duration of a task (or only weak bounds) are known. Therefore, the scheduling is performed in practice by assigning a resource to a task at the moment some resource becomes available.
The question is: What is the minimal number of resources needed to guarantee the optimal execution time
achievable with an unlimited number of resources?

The answer is simple: since there is no information about the duration of tasks, every reachable marking of the workflow net without durations
may be also reached for some assignment of durations.
Let $M$ be a reachable marking with a maximal number of tokens, say $k$, in places with positive duration, and let $d_1 \leq d_2 \leq \cdots \leq d_k$ be the durations  of their associated tasks. If less than $k$ resources are available, and we do not assign a resource to the task with duration $d_k$, we introduce a delay with respect to the case of an unlimited number of resources. On the contrary, if the number of available resources is $k$, then the scheduler for $k$ resources can always simulate the behaviour of the scheduler for an unlimited
number of resources.

\begin{definition}
Let $N=(P,T,F,I,O,\tau)$ be a workflow Petri net. For every marking $M$ of $N$, define the \emph{concurrency} of $M$ as 
$\conc{M} \defeq \sum_{p \in D} M(p)$, where $D \subseteq P$ is the set of places $p \in P$ such that $\tau(p) > 0$.
The \emph{concurrency threshold of $N$} is defined by
    \begin{align*}
        \CT{N} \defeq \max \set{ \conc{M} \mid M \in \reach[N]{M} }.
    \end{align*}
\end{definition}

The following lemma follows easily from the definitions.

\begin{lemma}
For every workflow net $N$: $\RT{N} \leq \CT{N}$.
\end{lemma}
\begin{proof}
Follows immediately from the fact that for every schedule $f$ of a run of $N$, there is a schedule $g$ with $\CT{N}$
machines such that $g(p) \leq f(p)$ for every place $p$.
\end{proof}

In the rest of the paper we study the complexity of computing the concurrency threshold.
In~\cite{BVT16}, it was shown that the threshold can be computed in polynomial time for regular workflows,
a class with a very specific structure, and the problem for the general free-choice case was left open.
In Section~\ref{ssec:marked-graphs-ct} we prove that the concurrency threshold of marked graphs can be computed in polynomial time by reduction to a linear programming problem over the rational numbers. In Section~\ref{ssec:free-choice-ct} we study the free-choice case. We show that deciding if the threshold exceeds a given value is NP-complete for acyclic, sound free-choice workflow nets. Further, it can be computed by solving the same linear programming problem as in the case of marked graphs, but over the integers. Finally, we show that in the cyclic case  the problem
remains NP-complete, but the integer linear programming problem does not necessarily yield the correct solution.

\subsection{Concurrency threshold of marked graphs}\label{ssec:marked-graphs-ct}
The concurrency threshold of marked graphs can be computed using a standard technique
based on the \emph{marking equation}~\cite{Murata89}. Given a net $N=(P,T,F)$,
define the \emph{incidence matrix} of $N$ as the $\abs{P} \times \abs{T}$ matrix $\vec{N}$ given by:

\[\vec{N}(p, t) =\left\{
\begin{array}{rl}
1 & \mbox{ if $p \in \postset{t} \setminus \preset{t}$} \\
-1 & \mbox{ if $p \in \preset{t} \setminus \postset{t}$} \\
0 & \mbox{ otherwise}
\end{array} \right.\]

In the following, we denote by $\vec{M}$ the representation
of a marking $M$ as a vector of dimension $\abs{P}$.
Let $N$ be a Petri net, and let $M_1, M_2$ be markings of $N$. The following results are well known from the literature (see e.g.~\cite{Murata89}):
\begin{itemize}
\item If $M_2$ is reachable from $M_1$ in $N$, then $\vec{M_2} = \vec{M_1} + \vec{N} \cdot \vec{X}$ for some integer vector $\vec{X} \ge 0$.
\item If $N$ is a marked graph and $\vec{M_2} = \vec{M_1} + \vec{N} \cdot \vec{X}$ for some \emph{rational} vector $\vec{X} \ge 0$, then $M_2$ is reachable from $M_1$ in $N$.
\item If $N$ is acyclic and $\vec{M_2} = \vec{M_1} + \vec{N} \cdot \vec{X}$ for some \emph{integer} vector $\vec{X} \ge 0$, then $M_2$ is reachable from $M_1$ in $N$.
\end{itemize}
Given a workfkow net $N=(P,T,F,I,O,\tau)$, let 
$\vec{D} \colon P \mapsto \N$ be the vector defined by
$\vec{D}(p) = 1$ if $p \in D$ and $\vec{D}(p) = 0$ if $p \notin D$, where $D$ is the set of places with positive duration. We define the linear optimization problem
\begin{align}\label{eqn:linear-program}
   \ell^N =  \max\set{\vec{D} \cdot \vec{M} \mid \vec{M} = \vec{M_I} + \vec{N} \cdot \vec{X}, \vec{M} \ge 0, \vec{X} \ge 0}
\end{align}
Since the solutions of $\vec{M} = \vec{M_I} + \vec{N} \cdot \vec{X}$ contain all the reachable markings of $(N, M_I)$, we have $\ell^N \geq \CT{N}$. Further, using these results above, we obtain:
\begin{theorem}\label{thm:upperbound}
Let $N$ be a workflow net, and let $\ell^N_\Q$ and $\ell^N_\Z$ be the solution of the linear optimization problem~(\ref{eqn:linear-program}) over the rationals and over the integers, respectively. We have:
\begin{itemize}
\item $\ell^N_\Q \geq \ell^N_\Z \geq \CT{N}$;
\item If $N$ is a marked graph, then $\ell_\Q=\ell_\Z = \CT{N}$.
\item If $N$ is acyclic, then $\ell_\Q \geq \ell_\Z = \CT{N}$.
\end{itemize}
\end{theorem}

In particular, it follows that $\CT{N}$ can be computed in polynomial time for marked graphs, acyclic or not. (The result about acyclic nets is used in the next section.)

\subsection{Concurrency threshold of free-choice nets}\label{ssec:free-choice-ct}
We study the complexity of computing the concurrency threshold of 
free-choice workflow nets. We first show that, contrary to numerous other properties for which there
are polynomial algorithms, deciding if the concurrency threshold exceeds a given value is NP-complete.

\begin{restatable}{theorem}{thmNP}\label{thm:NP2}
The following problem is NP-complete:
\begin{quote}
{\bf Given}: A sound, free-choice workflow net $N=(P,T,F,I,O)$, and a number $k \leq |T|$. \\
{\bf Decide}: Is the concurrency threshold of $N$ at least $k$?
\end{quote}
\end{restatable}
\begin{proof}
A detailed proof can be found in the appendix, here we only sketch the argument. 
Membership in NP is nontrivial, and follows from results of~\cite{Aalst97,DeselEsparza95}.
We prove NP-hardness by means of a reduction from Maximum Independent Set (MIS):
\begin{quote}
{\bf Given}: An undirected graph $G=(V, E)$, and a number $k \leq |V|$. \\
{\bf Decide}: Is there a set $\In \subseteq V$ such that $\abs{\In} \geq k$ and $\set{v,u} \notin E$ for every $u,v \in \In{}$?
\end{quote}

Given a graph $G=(V,E)$, we construct a sound free-choice workflow net $N_G$ in polynomial time as follows:
\begin{itemize}
    \item For each $e=\{v, u\} \in E$ we add to $N_G$ the ``gadget net'' $N_e$ shown in Fig.~\ref{fig:gadget-ne}, and for every node $v$ we add the gadget net $N_v$ shown in Fig.~\ref{fig:gadget-nv}.
    \item For every $e=\{v, u\} \in E$, we add an arc from the place $[e,v]^4$ of $N_e$ to the  transition $v^1$ of $N_v$, and from $[e,u]^4$ to the transition $u^1$ of $N_u$.
    \item The set $I$ of initial places contains the place $e^0$ of $N_e$ for every edge $e$; the set $O$ of output places contains the places $v^2$ of the nets $N_v$.
\end{itemize}

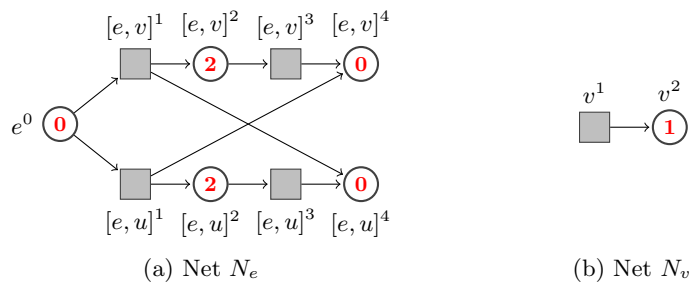
\begin{figure}
\centering


\subfloat[Net $N_e$]{\
    \begin{tikzpicture}[scale=1.0,every node/.style={scale=1.0}]
        \node [minimum width=1ex, minimum height=23ex] at (0,0) {};

        \node [place,label={left:$e^0$}] (e0) at (0,0) {\weight{0}};
        \node [place,label={above:$[e,v]^2$}] (ev2) at (2,0.8) {\weight{2}};
        \node [place,label={below:$[e,u]^2$}] (eu2) at (2,-0.8) {\weight{2}};
        \node [place,label={above:$[e,v]^4$}] (ev4) at (4,0.8) {\weight{0}};
        \node [place,label={below:$[e,u]^4$}] (eu4) at (4,-0.8) {\weight{0}};

        \node [transition,label={above:$[e,v]^1$}] (ev1) at (1,0.8) {}
          edge [pre]  (e0)
          edge [post] (ev2)
          edge [post] (eu4);
        \node [transition,label={below:$[e,u]^1$}] (eu1) at (1,-0.8) {}
          edge [pre]  (e0)
          edge [post] (eu2)
          edge [post] (ev4);
        \node [transition,label={above:$[e,v]^3$}] (ev3) at (3,0.8) {}
          edge [pre]  (ev2)
          edge [post] (ev4);
        \node [transition,label={below:$[e,u]^3$}] (eu3) at (3,-0.8) {}
          edge [pre]  (eu2)
          edge [post] (eu4);
    \end{tikzpicture}\label{fig:gadget-ne}
}\hspace{2cm}
\subfloat[Net $N_v$]{\
    \begin{tikzpicture}[scale=1.0,every node/.style={scale=1.0}]
        \node [minimum width=1ex, minimum height=23ex] at (6,0) {};

       \node [place,label={above:$v^2$}] (v2) at (6.5,0) {\weight{1}}; 
        \node [transition,label={above:$v^1$}] (v1) at (5.5,0) {}
          edge [post] (v2);
    \end{tikzpicture}\label{fig:gadget-nv}
}

\caption{Gadgets for the proof of Theorem~\ref{thm:NP2}.}\label{fig:gadgets}
\end{figure}

It is easy to see that $N_G$ is free-choice and sound, and in the Appendix we
show the result of applying the reduction to a small graph and
prove that $G$ has an independent set of size at least $k$ if{}f the concurrency threshold of $(N_G,M_I)$ is at least $2|E|+k$.
The intuition is that for each edge $e \in E$, we fire the transition $[e,u]^1$ where $u \notin \In{}$,
and for each $v \in \In{}$, we fire the transition $v^1$, thus marking one of $[e,u]^2$ or
$[e,v]^2$ for each edge $e \in E$ and the place $v^2$ for each $v \in \In{}$.
\end{proof}

\subsection{Approximating the concurrency threshold}
Recall that the solution of problem~(\ref{eqn:linear-program}) over the rationals or the integers is always an upper bound on the concurrency threshold for any Petri net (Theorem~\ref{thm:upperbound}). The question is whether any stronger result holds when the workflows are sound and free-choice. Since computing the concurrency threshold is NP-complete, we cannot expect the solution over the rationals, which is computable in polynomial time, to provide the exact value. However, it could still be the case that the solution over the integers is always exact. Unfortunately, this is not true, and we can prove the following results:

\begin{theorem}
Given a Petri net $N$, let $\ell^N_\Q$ and $\ell^N_\Z$ be as in Theorem~\ref{thm:upperbound}.
\begin{itemize}
\item[(a)] There is an acyclic sound free-choice workflow net $N$ such that $\CT{N}  < \ell^N_\Q$.
\item[(b)] There is a sound free-choice workflow net $N$ such that and let $\CT{N} < \ell^N_\Z$.
\end{itemize}
\end{theorem}
\begin{proof}
For (a), we can take the net obtained by adding to the gadget in Fig.~\ref{fig:gadget-ne} a new transition with input places $[e,v]^4$ and $[e,u]^4$, and an output place $o$ with weight $2$. We take $e^0$ as input place. The concurrency threshold is clearly $2$, reached, for example, after firing $[e,v]^1$. However, we have $\ell^N_\Q = 3$, reached by the rational solution $\vec{X} = (1/2, 1/2, \ldots, 1/2)$. Indeed, the marking equation then yields the marking $M$ satisfying $M([e,v]^2) = M([e,u]^2)=M(o)=1/2$.

For (b), we can take the workflow net of 
Fig.~\ref{fig:notexact}. It is easy to see that the concurrency threshold is equal to $1$. The marking $\vec{M}$ that puts one token in each of the 
two places with weight $1$, and no token in the rest of the places, is not reachable from $M_I$.
However, it is a solution of the marking equation, even when solved over the integers. Indeed, we
have $\vec{M} = \vec{M_I} + \vec{N} \cdot \vec{X}$ for $\vec{X}=(1,0,1,1,0,0,1)$. Therefore, the upper bound derived from the marking equation is $2$. 
\end{proof}

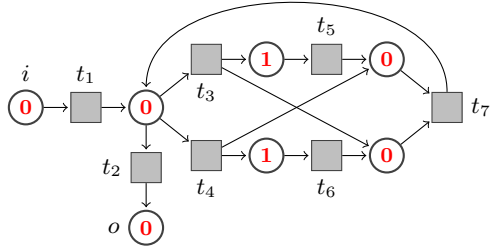
\begin{figure}
\centering

\begin{tikzpicture}[scale=0.8,every node/.style={scale=1.0}]

        \node [place,label={above:$i$}] (i) at (-2,0) {\weight{0}};
        \node [place] (e0) at (0,0) {\weight{0}};
        \node [place] (ev2) at (2,0.8) {\weight{1}};
        \node [place] (eu2) at (2,-0.8) {\weight{1}};
        \node [place] (ev4) at (4,0.8) {\weight{0}};
        \node [place] (eu4) at (4,-0.8) {\weight{0}};
        \node [place, label={left:$o$}] (o) at (0,-2) {\weight{0}};

        \node [transition,label={above:$t_1$}] (t1) at (-1,0) {}
         edge [pre]  (i)
         edge [post] (e0);
         \node [transition,label={left:$t_2$}] (t2) at (0,-1) {}
         edge [pre]  (e0)
         edge [post] (o);
        \node [transition,label={below:$t_3$}] (ev1) at (1,0.8) {}
          edge [pre]  (e0)
          edge [post] (ev2)
          edge [post] (eu4);
        \node [transition,label={below:$t_4$}] (eu1) at (1,-0.8) {}
          edge [pre]  (e0)
          edge [post] (eu2)
          edge [post] (ev4);
        \node [transition,label={above:$t_5$}] (ev3) at (3,0.8) {}
          edge [pre]  (ev2)
          edge [post] (ev4);
        \node [transition, label={below:$t_6$}] (eu3) at (3,-0.8) {}
          edge [pre]  (eu2)
          edge [post] (eu4);
        \node [transition, label={right:$t_7$}] (eo) at (5,0) {}
          edge [pre]  (eu4) edge [pre]  (ev4)
          edge [post,bend right=90] (e0);

\end{tikzpicture}
\caption{A sound free-choice workflow net for which the linear programming problem derived from the marking equation does not yield the exact value of the concurrency bound, even when solved over the integers.}\label{fig:notexact}
\end{figure}

\section{Concurrency threshold: a practical approach}\label{sec:experiments}
We have implemented a tool\footnote{The tool is available from \url{https://gitlab.lrz.de/i7/macaw}.}
to compute an upper bound on the concurrency threshold by constructing a linear program
and solving it by calling the mixed-integer linear programming solver Cbc from the COIN-OR project~\cite{Lougee-Heimer03}.
Additionally, fixing a number $k$, we used the
state-of-the art Petri net model checker LoLA~\cite{Wolf07} to both establish a lower bound,
by querying LoLA for existence of a reachable marking $M$ with $\conc{M} \ge k$;
and to establish an upper bound, by querying LoLA if all reachable markings $M'$ satisfy
$\conc{M'} \le k$.

We evaluated the tool on a set of 1386 workflow nets
extracted from a collection of five libraries of industrial business processes
modeled in the IBM WebSphere Business Modeler~\cite{FFJKLVW09}.
For the concurrency threshold, we set $D = P \setminus O$.
These nets also have multiple output places, however with a slightly
different semantics for soundness allowing unmarked output places in
the final marking. We applied the transformation described in~\cite{KHA03}
to ensure all output places will be marked in the final marking.
This transformation preserves soundness and the concurrency threshold.

All of the 1386 nets in the benchmark libraries are free-choice nets. We selected the sound
nets among them, which are 642. Out of those 642 nets, 409 are marked graphs. Out of
the remaining 233 nets, 193 are acyclic and 40 cyclic.
We determined the exact concurrency threshold of all sound nets
with LoLA using state-space exploration.
Fig.~\ref{fig:benchmarks-concurrency} shows the distribution of the threshold.

\begin{figure}
    \input{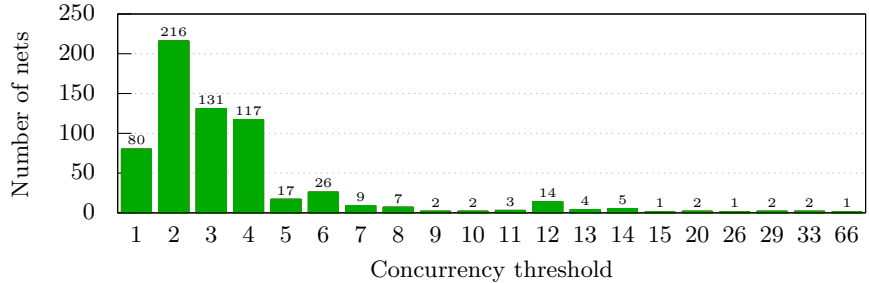}
    \caption{Distribution of the concurrency threshold of the 642 nets analyzed.}\label{fig:benchmarks-concurrency}
\end{figure}

On all 642 sound nets, we computed an upper bound on
the concurrency threshold using our tool,
both using rational and integer variables.
We computed lower and upper bounds using LoLA with the value $k = \CT{N}$
of the concurrency threshold. We report the results for computing the lower and upper bound separately.

All experiments were performed on the same machine equipped with
an Intel Core i7-6700K CPU and 32\,GB of RAM\@.
The results are shown in Table~\ref{tab:benchmarks-results}.
Using the linear program, we were able to compute an upper bound for all nets
in total in less than 7 seconds, taking at most 30 milliseconds
for any single net.
LoLA could compute the lower bound for all nets in 6 seconds.
LoLA fails to compute the upper bound in three cases
due to reaching the memory limit of 32\,GB\@.
For the remaining 639 nets, LoLA could compute the upper bound within
7 minutes in total.

We give a detailed analysis for the 9 nets with a state space of over one million.
For three nets with state space of sizes $10^9$, $10^{10}$ and $10^{17}$,
LoLa reaches the memory limit.
For four nets with state spaces between $10^6$ and $10^8$ and concurrency
threshold above 25, LoLA takes
2, 10, 48 and 308 seconds each.
For two nets with a state space of $10^8$ and a concurrency threshold of just 11,
LoLA can establish the upper bound in at most 20 milliseconds.
The solution of the linear program can be computed in all 9 cases in less than
30 milliseconds.

\begin{table}
    \setlength{\tabcolsep}{3pt}
    \begin{center}
    \begin{tabular}[t]{lrrrrrrrrr}
        \toprule
        & \multicolumn{3}{c}{Net size} & & \multicolumn{4}{c}{Analysis time (sec)} \\
        \cmidrule(r){2-4}
        \cmidrule(r){6-9}
        & $\abs{P}$ & $\abs{T}$ & $\abs{\reachwn[N]}$ & $\CT{N}$ & $\ell^N_\Q$ &  $\ell^N_\Z$ & $\CT{N} \ge k$ & $\CT{N} \le k$ & \\
        \midrule
        Median &  21\phantom{.0}   &  14\phantom{.0}  &  16\phantom{$\ \,\cdot10^{00}$}              &
                   3\phantom{.0}   &   0.01           &   0.01 & 0.01    &   0.01\phantom{\textsuperscript{*}}\\
        Mean   &  28.4             &  18.6            &  $3\cdot10^{14}$ &
                   3.7             &   0.01           &   0.01 & 0.01    &   0.58\textsuperscript{*} \\
        Max    & 262\phantom{.0}   & 284\phantom{.0}  &  $2\cdot10^{17}$ &
                  66\phantom{.0}   &   0.03           &   0.03 & 1.18    & 307.76\textsuperscript{*}  \\
        \bottomrule
    \end{tabular}
    \end{center}
    \caption{Statistics on the size and analyis time for the 642 nets analyzed.
    The times marked with \textsuperscript{*} exclude the 3 nets where LoLA reaches the memory limit.}\label{tab:benchmarks-results}
\end{table}

Comparing the values of the upper bound, first
we observed that we obtained the same value using either rational or integer variables.
The time difference between both was however negligible.
Second, quite surprisingly, we noticed that the upper bound obtained from the linear program
is exact in all of our cases, even for the cyclic ones.
Further, it can be computed much faster in several cases than the upper bound
obtained by LoLA and it gives a bound in all cases,
even when the state-space exploration reaches its limit.
By combining linear programming for the upper bound and state-space exploration for
the lower bound, an exact bound can always be computed within a few seconds.

\section{Conclusion}\label{sec:conclusion}
Planning sufficient execution resources for a business or production process is a crucial part of process engineering~\cite{LiuZLQ14,BessaiYOGN12,XuLZ08}. We considered a simple version of this problem in which resources are uniform and tasks are not interruptible. We studied the complexity of computing the resource threshold, i.e.,~the minimal number of resources
allowing an optimal makespan.
We showed that deciding if the resource threshold exceeds a given bound is NP-hard even for acyclic marked graphs. For this reason, we investigated the complexity of computing the concurrency threshold, an upper bound of the resource threshold introduced in~\cite{BVT16}. Solving a problem left open in~\cite{BVT16}, we showed that
deciding if the concurrency threshold exceeds a given bound is NP-hard for general sound free-choice workflow nets.
We then presented a polynomial-time approximation algorithm, and showed experimentally that it computes the \emph{exact} value of the concurrency threshold for all benchmarks of a standard suite of free-choice workflow nets.

\bibliographystyle{splncs03}
\bibliography{references}

\clearpage
\appendix

\section*{Appendix}
\label{sec:appendix}
\thmNPMG*

\begin{proof}
Since $N$ is acyclic, every transition can occur at most once in a firing sequence. Further, since $N$ is a sound marked graph, it has exactly one run, which is isomorphic to $N$ (This follows immediately from the fact that nonsequential processes are  sound marked graphs themselves). Therefore, in order to decide if $\RT{N} \leq k$ holds in 
nondeterministic polynomial time, we just guess the starting time $f(p)$ of every place $p$ of $N$, and check that
it satisfies conditions (1)-(3) in the definition of resource threshold.

To prove NP-hardness, let $J_1, \ldots, J_n$, $\tau$, $t>0$, and $k$ be an instance of the scheduling problem above. We construct an acyclic  workflow marked graph $N=(P,T,F, \tau')$ and a number $k'$ such that $\RT{N} \leq k'$ if{}f the jobs can be executed in time $t$ with $k$ machines. The procedure for the construction of $N$ is as follows:
\begin{itemize}
\item Add input and output places $p_I, p_O$, transitions $t_I, t_O$, arcs $(p_I, t_I), (t_O, p_O)$.
\item For every job $J_i$ minimal w.r.t. $\preceq$: add a place $p_i$, a transition $t_i'$, and arcs  $(t_I, p_i), (p_i, t_i')$.
\item For every job $J_i$ maximal w.r.t. $\preceq$: add a place $p_i$, a transition $t_i$, and an arc  $(t_i, p_i), (p_i, t_O)$.
\item For every other job $J_i$ add a place $p_i$, transitions $t_i, t_i'$, and arcs $(t_i, p_i), (p_i, t_i')$.
\item For every pair $J_i \preceq J_j$, add a place $p_{i, j}$ and arcs $(t_i', p_{i,j}), (p_{i,j}, t_j)$.
\item Set $\tau'(p_i) = \tau(J_i)$ for every place $p_i$ and $\tau'(p') = 0$ for every other place $p'$.
\end{itemize}
Let $N$ be the net constructed so far. Now test if $t_{\min}(N) \le t$.
If $t_{\min}(N) > t$, then there is no execution of $J_1, \ldots, J_n$ in time $t$ with any amount of machines, so we force
a negative answer (e.g.~by setting $k' = 0$).
Otherwise, we add a place $p_t$ and two arcs $(t_I, p_t), (p_t, t_O)$ and set $\tau'(p_t)= t$.
Then $t_{\min}(N) = t$ and we set $k'=k+1$.
By the definition of $N$, we have that $J_1, \ldots, J_n$ can be executed in time $t$ with
$k$ machines if{}f $N$ can be executed in time $t$ with $k' = k+1$ resources.
\end{proof}

\thmNP*

\begin{proof}
Let $N=(P,T,F,I,O)$ be a sound free-choice workflow net. It is well known that the system
$(\overline{N}, M_I)$, where $\overline{N}$ is the result of adding to $N$ a new transition
$\overline{t}$ with preset $O$ and postset $I$ is live and bounded (This is shown in~\cite{Aalst97} for workflow nets with one input and one output place, and the same proof can be used in the general case as well.) In particular, every firing sequence $\sigma \overline{t}$ of $(\overline{N}, M_I)$ leads back to the initial marking $M_I$, since otherwise it would lead to a marking $M' \geq M_I$ such that $M'\neq M_I$, and so $(\overline{N}, M_I)$ would not be bounded. It follows that every reachable marking of $(\overline{N}, M_I)$ is reachable by means of a firing sequence that does not contain $\overline{t}$, and so in particular $(N, M_I)$ and $(\overline{N}, M_I)$ have the same concurrency threshold.

Let $M_{\max}$ be a marking witnessing the concurrency threshold of $(\overline{N}, M_I)$.
By the Shortest Sequence Theorem~\cite{DeselEsparza95}, there exists a firing sequence $M_I \trans{\sigma} M_{\max}$ such that $|\sigma| \leq n (n+1)(n+2)/6$, where $n=|T|$. Therefore, the  nondeterministic algorithm that
guesses a firing sequence of length at most $n (n+1)(n+2)/6$ steps, and halts if the current marking has at least $k$ tokens in the places of $D$, runs in polynomial time.

We prove NP-hardness by means of a reduction from Maximum Independent Set (MIS):
\begin{quote}
{\bf Given}: An undirected graph $G=(V, E)$, and a number $k \leq |V|$. \\
{\bf Decide}: Is there a set $\In \subseteq V$ such that $\abs{\In} \geq k$ and $\set{v,u} \notin E$ for every $u,v \in \In{}$?
\end{quote}
We illustrate the reduction on the instance of MIS shown in Fig.~\ref{fig:red-graph}, where $\{v_1, v_4, v_5\}$ is a maximum independent set.
Fig.~\ref{fig:red-net} shows the result of the reduction.
Here we represent the set $D$ by weights of places as for the equation~(\ref{eqn:linear-program}).
We use the weight $2$ for a place $p$ as a compact representation of two places with weight 1 each
and the same pre- and postset as $p$.

Given a graph $G=(V,E)$, we construct a sound free-choice workflow net $N_G$ in polynomial time as follows:
\begin{itemize}
    \item For each $e=\{v, u\} \in E$ we add to $N_G$ the ``gadget net'' $N_e$ shown in Fig.~\ref{fig:gadget-ne2}, and for every node $v$ we add the gadget net $N_v$ shown in Fig.~\ref{fig:gadget-nv2}.
    \item For every $e=\{v, u\} \in E$, we add an arc from the place $[e,v]^4$ of $N_e$ to the  transition $v^1$ of $N_v$, and from $[e,u]^4$ to the transition $u^1$ of $N_u$ (see Fig.~\ref{fig:red-net}).
    \item The set $I$ of initial places contains the place $e^0$ of $N_e$ for every edge $e$; the set $O$ of output places contains the places $v^2$ of the nets $N_v$.
\end{itemize}

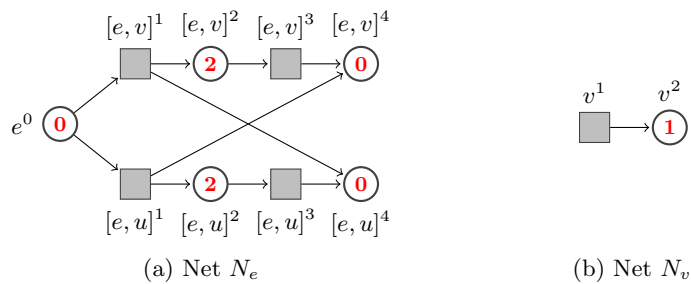
\begin{figure}
\centering


\subfloat[Net $N_e$]{\
    \begin{tikzpicture}[scale=1.0,every node/.style={scale=1.0}]
        \node [minimum width=1ex, minimum height=23ex] at (0,0) {};

        \node [place,label={left:$e^0$}] (e0) at (0,0) {\weight{0}};
        \node [place,label={above:$[e,v]^2$}] (ev2) at (2,0.8) {\weight{2}};
        \node [place,label={below:$[e,u]^2$}] (eu2) at (2,-0.8) {\weight{2}};
        \node [place,label={above:$[e,v]^4$}] (ev4) at (4,0.8) {\weight{0}};
        \node [place,label={below:$[e,u]^4$}] (eu4) at (4,-0.8) {\weight{0}};

        \node [transition,label={above:$[e,v]^1$}] (ev1) at (1,0.8) {}
          edge [pre]  (e0)
          edge [post] (ev2)
          edge [post] (eu4);
        \node [transition,label={below:$[e,u]^1$}] (eu1) at (1,-0.8) {}
          edge [pre]  (e0)
          edge [post] (eu2)
          edge [post] (ev4);
        \node [transition,label={above:$[e,v]^3$}] (ev3) at (3,0.8) {}
          edge [pre]  (ev2)
          edge [post] (ev4);
        \node [transition,label={below:$[e,u]^3$}] (eu3) at (3,-0.8) {}
          edge [pre]  (eu2)
          edge [post] (eu4);
    \end{tikzpicture}\label{fig:gadget-ne2}
}\hspace{2cm}
\subfloat[Net $N_v$]{\
    \begin{tikzpicture}[scale=1.0,every node/.style={scale=1.0}]
        \node [minimum width=1ex, minimum height=23ex] at (6,0) {};

       \node [place,label={above:$v^2$}] (v2) at (6.5,0) {\weight{1}}; 
        \node [transition,label={above:$v^1$}] (v1) at (5.5,0) {}
          edge [post] (v2);
    \end{tikzpicture}\label{fig:gadget-nv2}
}

\caption{Gadgets for the proof of Theorem~\ref{thm:NP2}.}\label{fig:gadgets2}
\end{figure}

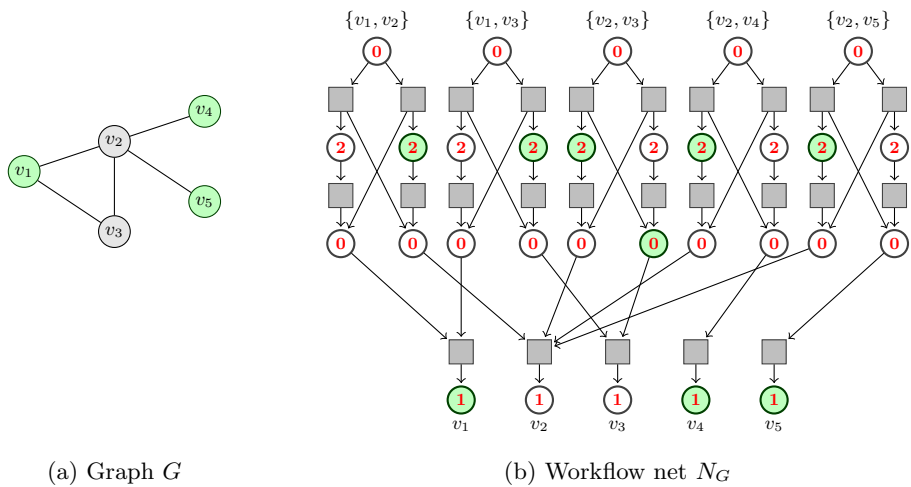
\begin{figure}
\centering


\subfloat[Graph $G$]{\
    \begin{tikzpicture}[scale=0.8,every node/.style={scale=0.8}]
        \node [minimum width=1ex, minimum height=55ex] at (-4,-2.85) {};

    \tikzstyle{vertex}=[circle, draw, fill=gray!20, inner sep=0pt, minimum size=15pt]
    \tikzstyle{weight} = [font=\small]
    \tikzstyle{edge} = [draw,thick,-]

    \node[vertex,marked] (v_1) at (-6,-2)     {$v_1$};
    \node[vertex] (v_2) at (-4.5,-1.5) {$v_2$};
    \node[vertex] (v_3) at (-4.5,-3)   {$v_3$};
    \node[vertex,marked] (v_4) at (-3,-1)     {$v_4$};
    \node[vertex,marked] (v_5) at (-3,-2.5)   {$v_5$};

    \foreach \source/ \dest /\weight in {v_1/v_2/e_1, v_1/v_3/e_2,  v_2/v_4/e_3, v_2/v_5/e_4}
        \draw (\source) -- node[weight,above] {} (\dest);

        \draw (v_2) -- node[font=\small,right] {} (v_3);

    \end{tikzpicture}\label{fig:red-graph}
}\hfill
\subfloat[Workflow net $N_G$]{\
    \begin{tikzpicture}[scale=0.8,every node/.style={scale=0.8}]
        \node [minimum width=1ex, minimum height=55ex] at (0,-2.85){};

\def\st{0}
    \node [place,label=above:{$\set{v_1, v_2}$}] (e10) at (0+\st,0) {\weight{0}};
        \node [place,marked] (e1v2) at (0.6+\st,-1.6) {\weight{2}};
        \node [place] (e1u2) at (-0.6+\st,-1.6) {\weight{2}};
        \node [place] (e1v4) at (0.6+\st,-3.2) {\weight{0}};
        \node [place] (e1u4) at (-0.6+\st,-3.2) {\weight{0}};

        \node [transition] (e1v1) at (0.6+\st,-0.8) {}
          edge [pre]  (e10)
          edge [post] (e1v2)
          edge [post] (e1u4);
        \node [transition] (e1u1) at (-0.6+\st,-0.8) {}
          edge [pre]  (e10)
          edge [post] (e1u2)
          edge [post] (e1v4);
        \node [transition] (e1v3) at (0.6+\st,-2.4) {}
          edge [pre]  (e1v2)
          edge [post] (e1v4);
        \node [transition] (e1u3) at (-0.6+\st,-2.4) {}
          edge [pre]  (e1u2)
          edge [post] (e1u4);

\def\st{2}
    \node [place,label=above:{$\set{v_1, v_3}$}] (e20) at (0+\st,0) {\weight{0}};
        \node [place,marked] (e2v2) at (0.6+\st,-1.6) {\weight{2}};
        \node [place] (e2u2) at (-0.6+\st,-1.6) {\weight{2}};
        \node [place] (e2v4) at (0.6+\st,-3.2) {\weight{0}};
        \node [place] (e2u4) at (-0.6+\st,-3.2) {\weight{0}};

        \node [transition] (e2v1) at (0.6+\st,-0.8) {}
          edge [pre]  (e20)
          edge [post] (e2v2)
          edge [post] (e2u4);
        \node [transition] (e2u1) at (-0.6+\st,-0.8) {}
          edge [pre]  (e20)
          edge [post] (e2u2)
          edge [post] (e2v4);
        \node [transition] (e2v3) at (0.6+\st,-2.4) {}
          edge [pre]  (e2v2)
          edge [post] (e2v4);
        \node [transition] (e2u3) at (-0.6+\st,-2.4) {}
          edge [pre]  (e2u2)
          edge [post] (e2u4);

\def\st{4}
    \node [place,label=above:{$\set{v_2, v_3}$}] (e30) at (0+\st,0) {\weight{0}};
        \node [place] (e3v2) at (0.6+\st,-1.6) {\weight{2}};
        \node [place,marked] (e3u2) at (-0.6+\st,-1.6) {\weight{2}};
        \node [place,marked] (e3v4) at (0.6+\st,-3.2) {\weight{0}};
        \node [place] (e3u4) at (-0.6+\st,-3.2) {\weight{0}};

        \node [transition] (e3v1) at (0.6+\st,-0.8) {}
          edge [pre]  (e30)
          edge [post] (e3v2)
          edge [post] (e3u4);
        \node [transition] (e3u1) at (-0.6+\st,-0.8) {}
          edge [pre]  (e30)
          edge [post] (e3u2)
          edge [post] (e3v4);
        \node [transition] (e3v3) at (0.6+\st,-2.4) {}
          edge [pre]  (e3v2)
          edge [post] (e3v4);
        \node [transition] (e3u3) at (-0.6+\st,-2.4) {}
          edge [pre]  (e3u2)
          edge [post] (e3u4);

\def\st{6}
    \node [place,label=above:{$\set{v_2, v_4}$}] (e40) at (0+\st,0) {\weight{0}};
        \node [place] (e4v2) at (0.6+\st,-1.6) {\weight{2}};
        \node [place,marked] (e4u2) at (-0.6+\st,-1.6) {\weight{2}};
        \node [place] (e4v4) at (0.6+\st,-3.2) {\weight{0}};
        \node [place] (e4u4) at (-0.6+\st,-3.2) {\weight{0}};

        \node [transition] (e4v1) at (0.6+\st,-0.8) {}
          edge [pre]  (e40)
          edge [post] (e4v2)
          edge [post] (e4u4);
        \node [transition] (e4u1) at (-0.6+\st,-0.8) {}
          edge [pre]  (e40)
          edge [post] (e4u2)
          edge [post] (e4v4);
        \node [transition] (e4v3) at (0.6+\st,-2.4) {}
          edge [pre]  (e4v2)
          edge [post] (e4v4);
        \node [transition] (e4u3) at (-0.6+\st,-2.4) {}
          edge [pre]  (e4u2)
          edge [post] (e4u4);
        
\def\st{8}
    \node [place,label=above:{$\set{v_2, v_5}$}] (e50) at (0+\st,0) {\weight{0}};
        \node [place] (e5v2) at (0.6+\st,-1.6) {\weight{2}};
        \node [place,marked] (e5u2) at (-0.6+\st,-1.6) {\weight{2}};
        \node [place] (e5v4) at (0.6+\st,-3.2) {\weight{0}};
        \node [place] (e5u4) at (-0.6+\st,-3.2) {\weight{0}};

        \node [transition] (e5v1) at (0.6+\st,-0.8) {}
          edge [pre]  (e50)
          edge [post] (e5v2)
          edge [post] (e5u4);
        \node [transition] (e5u1) at (-0.6+\st,-0.8) {}
          edge [pre]  (e50)
          edge [post] (e5u2)
          edge [post] (e5v4);
        \node [transition] (e5v3) at (0.6+\st,-2.4) {}
          edge [pre]  (e5v2)
          edge [post] (e5v4);
        \node [transition] (e5u3) at (-0.6+\st,-2.4) {}
          edge [pre]  (e5u2)
          edge [post] (e5u4);

       \node [place,marked,label=below:$v_1$] (v12) at (1.4,-5.8) {\weight{1}}; 
        \node [transition] (v11) at (1.4,-5) {}
          edge [post] (v12)
          edge [pre] (e1u4) edge [pre] (e2u4);
          
       \node [place,label=below:$v_2$] (v22) at (2.7,-5.8) {\weight{1}}; 
        \node [transition] (v21) at (2.7,-5) {}
          edge [post] (v22)
          edge [pre] (e1v4) edge [pre] (e3u4) edge [pre] (e4u4) edge [pre] (e5u4);

	 \node [place,label=below:$v_3$] (v32) at (4,-5.8) {\weight{1}}; 
        \node [transition] (v31) at (4,-5) {}
          edge [post] (v32)
          edge [pre] (e2v4) edge [pre] (e3v4);

      \node [place,marked,label=below:$v_4$] (v42) at (5.3,-5.8) {\weight{1}}; 
        \node [transition] (v41) at (5.3,-5) {}
          edge [post] (v42)
          edge [pre] (e4v4);
          
      \node [place,marked,label=below:$v_5$] (v52) at (6.6,-5.8) {\weight{1}}; 
      \node [transition] (v51) at (6.6,-5) {}
          edge [post] (v52)
          edge [pre] (e5v4);

    \end{tikzpicture}\label{fig:red-net}
}

\caption{A graph and its corresponding workflow net. A maximum independent set
and a marking achieving the concurrency threshold are highlighted in green.}\label{fig:redMIS}
\end{figure}

It is easy to see that $N_G$ is free-choice and sound. For the latter, observe that after firing either $[e,v]^1$ and $[e,v]^3$, or $[e,u]^1$ and $[e,u]^3$, for every edge $e \in E$, the same marking is reached, namely the one putting a token on all the places of the form $[x,y]^4$. From this marking 
we can then fire all the $v_i^1$ transitions to reach the final marking.

We claim that $G$ has an independent set of size at least $k$ if{}f the concurrency threshold of $(N_G,M_I)$ is at least $2|E|+k$.

For the first part, assume $G$ has an independent set $\In \subseteq V$. For every $e \in E$ let $v_e$ be a vertex of $e$ that does \emph{not} belong to $\In{}$.  Fire the transition $[e, v_e]^1$. Let $M$ be the marking so reached. Since $\In{}$ is an independent set, we have $M([e, v]^4) = 1$ for every $v \in \In{}$ and for every edge $e$ containing $v$. So $M$ enables $v^1$ for every $v \in \In{}$. Fire transition $v^1$ of $N_v$ for every $v \in \In{}$, reaching the marking $M'$. We then have $M'([e, v_e]^2)=1$ for every $e \in  E$, and $M'(v^2)=1$ for every $v \in \In{}$. So the concurrency threshold of $M'$ is at least $2|E| + |\In{}| \geq  2|E| + k$.

For the second part, assume that the concurrency threshold of $(N_G,M_I)$ is at least $2|E|+k$.
This can only be achieved by a marking $M$ marking exactly one of $[e, u]^2$ or $[e,  v]^2$ for every
$e = \set{u, v} \in E$, and additionally marking at least $k$ of the places $v^2$.
We claim that $\In = \set{ v \in V \mid M(v^2) = 1 }$ is an independent set.
Indeed, in a sequence reaching $M$, for every $e=\{v, u\} \in E$ only one of
$[e,v]^4$ or $[e,u]^4$ can become marked, but not both;
so no occurrence sequence leading to $M$ contains both $v^1$ and $u^1$,
and so $M$ does not mark both $v^2$ and $u^2$.
So $G$ has an independent set of size at least $k$.
\end{proof}

\end{document}